\documentclass[twocolumn, aps, pra, superscriptaddress]{revtex4-1}
\usepackage[T1]{fontenc}
\usepackage{amsmath}
\usepackage{appendix}
\usepackage[inline]{enumitem}
\usepackage{multirow}
\usepackage{amssymb}
\usepackage{amsthm}
\usepackage{caption}
\usepackage{graphicx}
\usepackage{bbm}
\usepackage{color}
\def \diracspacing {0.7pt}
\newcommand{\ket}[1]{| \hspace{\diracspacing} #1 \rangle} 
\newcommand{\braket}[2]{\langle #1 \hspace{\diracspacing} | \hspace{\diracspacing} #2 \rangle} 
\newcommand{\ketbra}[2]{| \hspace{\diracspacing} #1 \rangle \langle #2 \hspace{\diracspacing} |} 
\newcommand{\ketbraq}[1]{\ketbra{#1}{#1}} 
\newcommand{\norm}[2][]{#1\left| \! #1\left| #2 #1\right| \! #1\right|}

\newcommand{\abs}[2][]{#1| #2 #1|}
\newcommand{\smin}[0]{s_{\textnormal{min}}}
\newcommand{\smax}[0]{s_{\textnormal{max}}}

\newcommand{\unit}{\mathbb{I}}
\DeclareMathOperator{\tr}{tr}
\theoremstyle{definition}
\newtheorem{defn}{Definition}
\theoremstyle{plain}
\newtheorem{lem}[defn]{Lemma}
\newtheorem{thm}[defn]{Theorem}

\theoremstyle{remark}

\begin{document}
\title{Self-testing mutually unbiased bases in the prepare-and-measure scenario}
\author{M{\'a}t{\'e} Farkas}
\email{mate.farkas@phdstud.ug.edu.pl}
\affiliation{Institute of Theoretical Physics and Astrophysics, National Quantum Information Centre, Faculty of Mathematics, Physics and Informatics, University of Gdansk, 80-952 Gdansk, Poland}
\author{J\k{e}drzej Kaniewski}
\affiliation{Center for Theoretical Physics, Polish Academy of Sciences, Al.~Lotnik{\'o}w 32/46, 02-668 Warsaw, Poland}
\begin{abstract}
Mutually unbiased bases (MUBs) constitute the canonical example of incompatible quantum measurements. One standard application of MUBs is the task known as quantum random access code (QRAC), in which classical information is encoded in a quantum system, and later part of it is recovered by performing a quantum measurement. We analyse a specific class of QRACs, known as the $2^{d} \to 1$ QRAC, in which two classical dits are encoded in a $d$-dimensional quantum system. It is known that among rank-1 projective measurements MUBs give the best performance. We show (for every $d$) that this cannot be improved by employing non-projective measurements. Moreover, we show that the optimal performance can only be achieved by measurements which are rank-1 projective and mutually unbiased. In other words, the $2^{d} \to 1$ QRAC is a self-test for a pair of MUBs in the prepare-and-measure scenario. To make the self-testing statement robust we propose measures which characterise how well a pair of (not necessarily projective) measurements satisfies the MUB conditions and show how to estimate these measures from the observed performance. Similarly, we derive explicit bounds on operational quantities like the incompatibility robustness or the amount of uncertainty generated by the uncharacterised measurements. For low dimensions the robustness of our bounds is comparable to that of currently available technology, which makes them relevant for existing experiments. Lastly, our results provide essential support for a recently proposed method for solving the long-standing existence problem of MUBs.
%
%
\end{abstract}
\maketitle
%
%
\section{Introduction}
Mutually unbiased bases (MUBs) play an important role in many quantum
information processing tasks. They are optimal for quantum state determination
\cite{ivonovic,wooters_fields}, information locking \cite{locking1,locking2},
and the mean king's problem \cite{mean_king1,mean_king2}. Moreover, they give
rise to the strongest entropic uncertainty relations (among projective measurements)
\cite{maassen_uffink,uncertainty_survey,uncertainty_survey2}.
One intuitive way to look at them is the following: imagine that we encode a
classical message in a pure state corresponding to an element of a
basis. Then, if we measure this state in a basis unbiased to the initial
one, each measurement outcome occurs with the same probability. That is, we
do not learn anything about the originally encoded message.
Formally, two bases $\{\ket{a_i}\}_{i=1}^d$ and $\{\ket{b_j}\}_{j=1}^d$ in
$\mathbb{C}^d$ are mutually unbiased if
\begin{equation}
\label{eq:MUB}
\abs{ \braket{a_i}{b_j} }^{2} = \frac{1}{d} \quad \forall i, j \in [d] := \{ 1, 2, \ldots, d\}.
\end{equation}

Due to their importance, significant effort has been dedicated to investigating their structure
(see \cite{onMUBs} for a survey and \cite{MUB2-5} for a classification in
dimensions 2--5). It is known that in dimension $d$, there are at least 3 and at
most $d+1$ MUBs and the upper bound is saturated in prime power dimensions. The
maximal number of MUBs in composite dimensions is a long-standing open problem
(see \cite{MUB6_3,MUB6_4,MUB6_1,MUB6_2,MUB6_5,MUB6_6} for the case of dimension
6).

Another scenario in which MUBs perform well is the so-called
$2^{d} \to 1$ quantum random access code (QRAC) \cite{QRACSR,ArminQRAC}.
In this setup, two classical dits are
encoded into a qudit, and the aim is to recover one of them chosen
uniformly at random. It is well-known that sending a quantum system gives an
advantage over sending a classical system (of the same dimension) \cite{cRAC}
and this fact is used in many quantum information protocols~\cite{conjugate_coding,finite_automata1,complexity1,network1,locally1}.
%
It is commonly believed that the optimal performance of the $2^{d} \to 1$ QRAC is achieved when the measurements correspond to a pair of MUBs in dimension $d$, but this claim has only been proven for a restricted class of measurements~\cite{QRACMUB}.

The observation that quantum systems can give rise to stronger-than-classical correlations was first made by Bell~\cite{Bell} (although in a slightly different setup).
%
Moreover, it turns out that some of these strongly non-classical correlations can be achieved in an essentially {\em unique} manner. That is, the observed statistics allow us to identify
the employed states and measurements (up to local
isometries and extra degrees of freedom). The most prominent example of this kind
is the well-known CHSH inequality \cite{clauser69a}, which is
maximally violated by a pair of MUBs in dimension 2 on both sides 
\cite{tsirelson87a, chsh_selftest1,chsh_selftest2,chsh_selftest3}. Whenever
such an inference --- characterising the state and/or measurements based solely on
the observed statistics --- can be made, it is referred to as {\em
self-testing}
\cite{selftest0,selftest1,selftest2}. Self-testing is closely related to the concept of device-independent (DI) quantum information
processing, in which the devices used in the protocol are a priori untrusted~\cite{barrett05b, acin06a, colbeck06a, DI1, DI2}.
It is clear that what makes DI cryptography possible is precisely the
self-testing character of the correlations observed during the protocol.
By now self-testing is a well-developed field
\cite{selftest_state2, robust2, robust3, robust4, wang16a, supic16a, selftest_state1, selftest_state5} and includes results which are robust to noise~\cite{robust1, bancal15a, robust5, selftest_state3, selftest_state4, robust6, robust7}. Such statements are of particular interest, as they can be directly applied to experiments~\cite{selftest_exp}.

Recently the notion of self-testing has been extended to prepare-and-measure
scenarios~\cite{self-test_prepare}. In this setup, a preparation device creates
one of many possible quantum states and then sends it to a measurement
device. The latter performs one of many possible measurements on the state, and
then produces a classical output. This scenario encompasses many important quantum communication protocols, e.g.~the BB84 and B92 quantum
cryptography protocols \cite{BB84,B92}, and the aforementioned QRACs.

In the prepare-and-measure scenario one cannot distinguish between classical and quantum systems, unless additional restrictions are imposed.
The standard choice is to place an upper bound on the dimension of the system transmitted between the devices~\cite{SDI1,SDI2,SDI3}. This is often referred to as the semi-device-independent
(SDI) model for which several cryptographic protocols have been proposed~\cite{SDIprot1,SDIprot2,SDIprot3}.
In analogy to the DI model, it is clear that the security of SDI protocols is related to self-testing results in the prepare-and-measure scenario.

In this paper, we investigate the self-testing properties of the 
$2^{d} \to 1$ QRAC. In \cite{self-test_prepare}, the authors derive robust
self-testing results for $d = 2$ and ask whether similar statements hold for larger $d$.
We resolve this question by deriving a robust self-testing statement for arbitrary $d$. We show that the optimal performance in the $2^{d} \to 1$ QRAC certifies that the two measurements correspond to MUBs. To make the statement robust we propose new measures which characterise how close a pair of POVMs is to the MUB arrangement and derive explicit bounds on those in terms of the QRAC performance. Finally, we use this characterisation to obtain explicit bounds on operationally relevant quantities like the incompatibility robustness \cite{incomp} or the
amount of uncertainty produced. 

%
\section{Setup}
In the $2^{d} \to 1$ QRAC scenario (see Fig.~\ref{fig:qrac}), on the preparation
side Alice gets two uniformly random inputs, $i, j \in [d]$. Based on these
inputs she prepares a $d$-dimensional state $\rho_{ij}$, and sends it
to Bob who is on the measurement side. He gets a uniformly random input $y \in
\{1, 2\}$, which tells him which of Alice's inputs he is supposed to
guess. If $y = 1$, he aims to guess $i$, otherwise $j$.  This is performed by a
measurement on $\rho
_{ij}$, which we describe by the operators $\{A_i\}_i$
for $y = 1$, and $\{B_j\}_j$ for $y = 2$, where $A_i,B_j\ge0$,
$\sum_i{A_i}= \sum_j{B_j}=\unit$ and $i,j\in[d]$. The outcome of the measurement determines Bob's guess and the figure of merit is the \emph{average success probability} (ASP), which can be
written, using the above notation, as
\begin{equation}
\label{eq:aspform}
\bar{p} = \frac{1}{2 d^{2} } \sum_{ij} \tr \big[ \rho_{ij} ( A_{i} + B_{j} ) \big].
\end{equation}
%
%
%
%
%
\begin{figure}[h!]
%
\includegraphics[height=3cm]{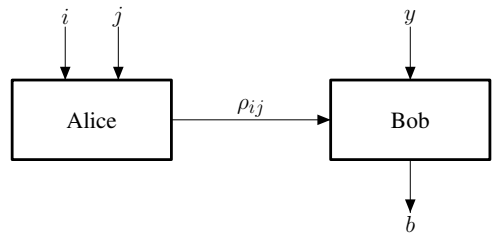}
\caption{Schematic representation of the $2^{d} \to 1$ QRAC protocol.}
%
\label{fig:qrac}
\end{figure}

\section{Ideal self-test}
To obtain the ideal self-testing statement we derive an achievable upper bound on the ASP and identify situations in which all the steps in the derivation are tight.
Note that $\tr \big[ \rho_{ij} ( A_{i} + B_{j} ) \big] \le \norm{ A_{i} + B_{j}
}$, where $\norm{.}$ is the operator norm, and since $(A_{i} + B_{j}) \ge 0$,
one can always find a state $\rho_{ij}$ such that this inequality is saturated. Let us from now on assume that the preparations are always chosen optimally, which allows us to focus solely on the measurements. Finding the maximal ASP can be performed using operator norm inequalities and other tools from matrix analysis, and yields the following theorem.
\begin{thm}
\label{thm:idealselftest}
The average success probability of the \mbox{$2^{d} \to 1$} QRAC is upper bounded by
\begin{equation}
\label{eq:aspbound}
\bar{p} \le \frac{1}{2} \left( 1 + \frac{1}{\sqrt{d}} \right) =: \bar{p}_{Q},
\end{equation}
and this bound can only be attained if Bob's measurements are rank-1 projective
and mutually unbiased. Moreover, in the optimal case the prepared states are the unique eigenstates of $A_{i} + B_{j}$, corresponding to the highest
eigenvalue.
\end{thm}
It was previously known that this upper bound holds if we restrict ourselves to rank-1 projective measurements and that among these measurements only MUBs can actually achieve it \cite{QRACMUB}. What we show is that the QRAC performance cannot be improved by employing non-projective measurements and that the optimal performance indeed requires MUBs, even if we allow for generic measurements. Note that this does not follow from any extremality
argument, as in general projective measurements are not the only
extremal $d$-outcome measurements~\cite{POVMsimulating}.

For a complete proof, we refer the reader to Appendix A.
Here, we state that the crucial step is to use
operator norm inequalities to show that the ASP is bounded by
\begin{equation}
\label{eq:aspboundt}
\bar{p} \leq \frac{1}{2} + \frac{1}{2 d^{2}} \sum_{ij} \sqrt{ t_{ij} },
\end{equation}
where $t_{ij} := \tr(A_iB_j)\ge0$, and therefore $\sum_{ij}t_{ij}=d$.
The right-hand side is strictly Schur-concave in $\{ t_{ij} \}_{ij}$, and hence is uniquely maximised by
the uniform distribution, $t_{ij}=\frac1d$ \cite{inequalities}, which yields $\bar{p}_Q$. A separate argument implies that to reach $\bar{p}_Q$ both measurements must be rank-1 projective and combining these two facts leads to the conclusion that the two measurements correspond to MUBs.

Theorem~\ref{thm:idealselftest} implies that the $2^{d} \to 1$ QRAC is an SDI
self-test for a pair of MUBs in dimension $d$: observing the optimal ASP implies that the two measurements constitute a pair of MUBs.
One might wonder whether the self-testing statement can be made even stronger, in the sense of providing more details about the measurements, but this is not possible. It is easy to check that every pair of MUBs is capable of producing the optimal ASP.
%
This ideal self-test is crucial for the success of the methods described in \cite{QRACMUB},
as there it is essential that the optimal QRAC ASP can only be obtained with an arbitrary pair of MUBs. 

\section{Robust self-test}
Since in a real experiment one never observes the optimal performance, the ideal
self-testing result is not sufficient. Instead, we need a \emph{robust}
self-testing statement, which tells us what can be certified in the case of
sub-optimal performance.

Inequality~\eqref{eq:aspboundt} implies that observing the optimal ASP forces
the distribution $\{ t_{ij} \}_{ij}$ to be uniform. For sub-optimal performance
we immediately get a bound on the $\frac{1}{2}$-R{\'e}nyi entropy ($H_
\frac12(\{q_i\}) = 2\log_2\left[\sum_i\sqrt{q_i}\right]$) of the distribution
$\{ t_{ij} / d \}_{ij}$, which we call the \emph{overlap entropy} $H_{S} (A, B) := H_{\frac{1}{2}} \big( \{ t_{ij} / d \}_{ij} \big)$. More concretely, 
from~\eqref{eq:aspboundt} we deduce that
\begin{equation}
\label{eq:H12bound}
H_{S}(A, B) \geq 2 \log_2 \big[ d \sqrt{d} ( 2 \bar{p} - 1 ) \big].
%
%
\end{equation}
This bound is non-trivial as long as $\bar{p}>\frac12+\frac{1}{2d\sqrt{d}}$ and observing $\bar{p} = \bar{p}_{Q}$ implies $H_{S}(A, B) = \log_2(d^2)$, which is the maximal value of the overlap entropy for a pair of POVMs. For $d = 4$ the lower bound is plotted in Fig.~\ref{fig:H12}.
\begin{figure}[h!]
\begin{center}
\includegraphics[width=0.9\columnwidth]{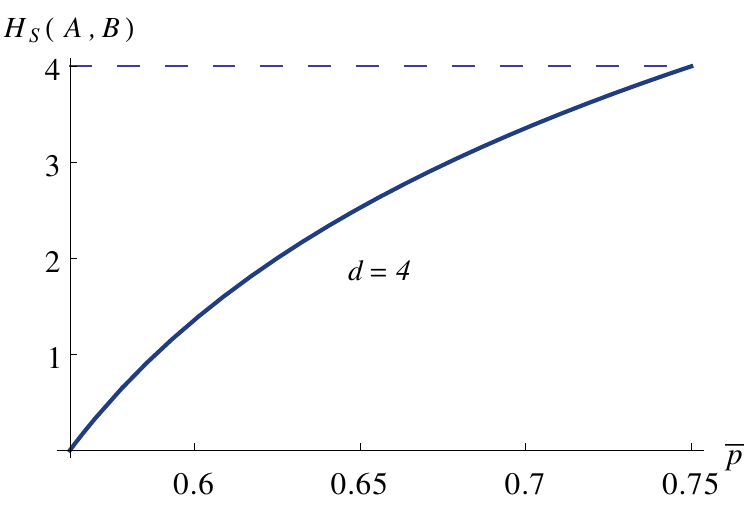}
\caption{Lower bound on the overlap entropy
for $\bar{p} \in[\frac12+\frac{1}{2d\sqrt{d}},\bar{p}_Q]$ in dimension 4.
\label{fig:H12}}
\end{center}
\end{figure}

Looking at the overlap entropy is not sufficient, because the maximal value can be achieved by measurements which are not MUBs, for instance the trivial measurements corresponding to $A_{i} = B_{j} = \unit/d$. The missing part is an argument showing that the measurements are close to being rank-1 projective. For a $d$-outcome measurement $\{ A_{i} \}_i$ acting on $\mathbb{C}^d$ this property can be assessed by looking at the sum of the norms, $N(A) := \sum_{i} \norm{ A_{i} }$, since for all measurements $N(A)\le d$ and the maximal value is attained if and only if the measurement is rank-1 projective. Therefore, saturating $N(A)=N(B)=d$ and
$H_S(A,B)=\log_2(d^2)$ certifies the MUB arrangement.

To obtain a bound on $N(A)$ we need a stronger version of Eq.~\eqref{eq:aspboundt}. In the Appendix B we show that
\begin{equation}
\label{eq:aspboundsn}
\bar{p} \leq \frac{1}{2} + \frac{1}{ 2 d^{2} } \sum_{ij} \big[ s_{ij} - ( 2 - \sqrt{2} ) s_{ij} n_{ij} \big],
\end{equation}
where $n_{ij} := 1 - \frac{1}{2} \big( \norm{ A_{i} } + \norm{ B_{j} } \big)$ and $s_{ij} := \norm{ \sqrt{A_{i} } \sqrt{ B_{j} } }$. This bound reduces to Eq.~\eqref{eq:aspboundt} if we omit the negative term and bound $s_{ij}$ by $\sqrt{t_{ij}}$, which constitutes an alternative derivation of Theorem~\ref{thm:idealselftest} (as $n_{ij}=0$ for all $i,j$ implies that both measurements are rank-1 projective).
%
%

The important feature of Eq.~\eqref{eq:aspboundsn} is that it allows us to lower bound the sum of the norms. In Appendix B we show that for $\bar{p} > \bar{p}_{0} := \frac{1}{2} + \frac{1}{2d^2} \sqrt{ ( d^2-1) d }$ we have
%
\begin{equation}
\label{eq:normsumbound}
N(A) \geq d - \frac{2 + \sqrt{2}}{d} \Big( 1 - \sqrt{ d^{3} ( 2 \bar{p} - 1 )^{2} -
( d^{2} - 1 ) } \Big)
\end{equation}
and by symmetry the same bound holds for $N(B)$. It is easy to check that for $\bar{p} = \bar{p}_{Q}$, the right-hand side evaluates to $d$, i.e.~the optimal performance certifies that both measurements are rank-1 projective. The lower bound given in Eq.~\eqref{eq:normsumbound} is plotted for $d = 4$ in Fig.~\ref{fig:normsum}.
\begin{figure}[h!]
\includegraphics[width=0.9\columnwidth]{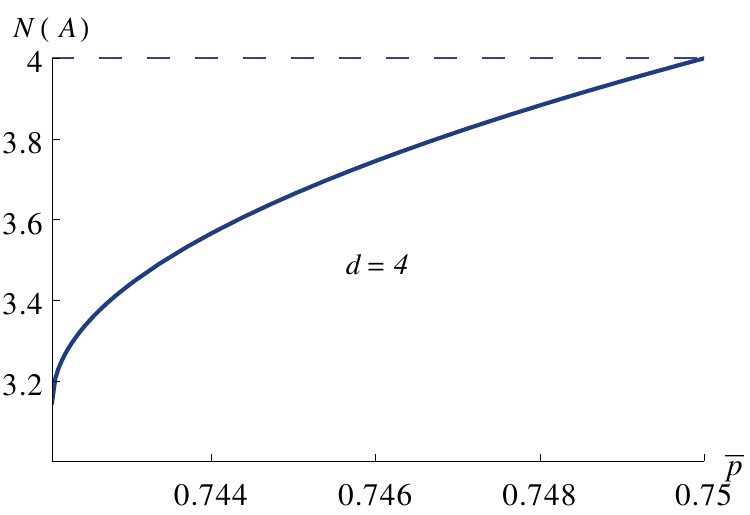}
\caption{Lower bound on the sum of the norms for $\bar{p} \in(\bar{p}_0,\bar{p}_Q]$
in dimension 4.}
\label{fig:normsum}
\end{figure}

Since Eqs.~\eqref{eq:normsumbound} and~\eqref{eq:H12bound} allow us to robustly certify the two defining properties of MUBs (rank-1 projectivity and uniformity of overlaps, respectively), combining them yields a robust self-test for MUBs. Note that the robustness is limited by Eq.~\eqref{eq:normsumbound} which requires that $\bar{p} > \bar{p}_{0}$.

%

\section{Operational bounds}
In the previous paragraph we have focused on quantities tailored to certify closeness to the MUB arrangement. Let us now show that a similar approach can be used to derive bounds on quantities which have an immediate operational meaning.
%

We begin with the incompatibility robustness. We say that two POVMs $\{ A_{i} \}_i$ and $\{ B_{j} \}_j$ are compatible 
(or jointly measurable) if there exists a parent POVM $\{M_{ij}\}_{ij}$, such
that $\sum_j M_{ij}=A_i$ and $\sum_i M_{ij}=B_j$  for all $i, j$. Otherwise they
are incompatible, which is often taken as the definition of non-classicality.
In order to quantify incompatibility beyond
this binary characterisation, the notion of incompatibility robustness has
been introduced \cite{incomp}. Consider the noisy POVMs, $A^\eta_i=\eta A_i+(1-\eta)\tr A_i \, \unit / d$, and similarly $B^\eta_j$. The incompatibility robustness $\eta^\ast$ of $A$ and $B$ is defined as the largest $\eta$ such that $\{A^\eta_i\}_i$ and $\{B^\eta_j\}_j$ are compatible. According to this measure MUBs are highly incompatible, but, perhaps surprisingly, they are not the most incompatible among rank-1 projective measurements in dimension $d$~\cite{mostincomp}.

Recently an analytic upper bound on $\eta^
{\ast}$ has been derived for an arbitrary set of POVMs~\cite{incompbound}. For a pair of POVMs the bound reads
\begin{equation}
\label{eq:incomp}
\eta^\ast\le\frac{d^2\max_{ij}\norm{A_i+B_j}-\sum_i\left(\tr A_i\right)^2
-\sum_j\left(\tr B_j\right)^2}{d\sum_i\tr A_i^2+d\sum_j
\tr B_j^2-\sum_i\left(\tr A_i\right)^2-\sum_j\left(\tr B_j \right)^2}.
\end{equation}
All the quantities appearing in this expression can be bounded using the previously developed methods, which leads to a bound which depends only on the observed performance $\bar{p}$. Since the final bound is rather complex, we do not present it here and refer the interested reader to Appendix C. The important feature of the bound is that for the optimal performance $\bar{p} = \bar{p}_{Q}$ we recover the correct value of the incompatibility robustness for a pair of MUBs, i.e.~$\eta^{\ast} = \frac{ \sqrt{d}/2 + 1}{ \sqrt{d} + 1 }$. In Fig.~\ref{fig:incomp} we plot the bound for $d = 4$ over the region where it is
non-trivial.

We note here that similar bounds can be derived for other measures of
incompatibility robustness using the same techniques. Among these is a measure
that uses arbitrary POVMs as noise \cite{Haa15}, for which MUBs are the most
incompatible pair of POVMs (of any number of outcomes) in dimension $d$
\cite{DFK19}. This can also be certified by observing $\bar{p} = \bar{p}_{Q}$.

%
%
%
%
\begin{figure}[h!]
\begin{center}
\includegraphics[width=0.9\columnwidth]{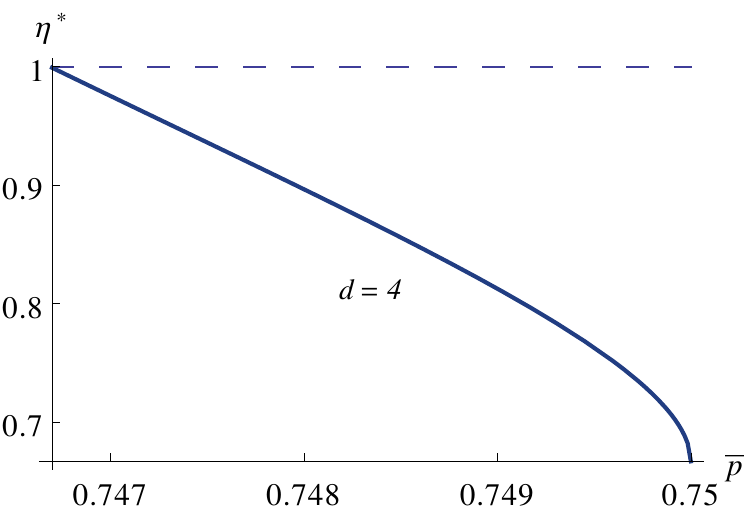}
\caption{Upper bound on the incompatibility robustness over the non-trivial region in dimension 4.
\label{fig:incomp}}
\end{center}
\end{figure}
%

The second operational quantity we consider is the amount of randomness produced by the uncharacterised measurements.
%
For a POVM $A$, let $H(A)_{\rho} := H \big( \big\{ \tr( A_{i} \rho ) \big\}_{i}
\big)$ be the Shannon entropy of the outcome statistics of $A$ on the state
$\rho$. Maassen and Uffink derived a state-independent lower bound on $H(A)_{\rho} + H(B)_{\rho}$ for rank-1 projective measurements~\cite{maassen_uffink}. For our purposes we need a more general statement which covers non-projective measurements. Such a bound has been derived in~\cite{entropic_POVM} and reads
%
%
\begin{equation}
\label{eq:uncertainty}
H(A)_{\rho} + H(B)_{\rho} \geq - \log_2 c,
\end{equation}
where $c := \max_{ij} \norm{ \sqrt{ A_{i} } \sqrt{ B_{j} } }^{2}$.
Therefore, we need an upper bound on $s_{ij}$ and such a bound has already been derived in Appendix B. The final statement reads
\begin{equation}\label{eq:uncertaintybound}
\begin{split}
&\left.H(A)_{\rho} + H(B)_{\rho} \geq \right. \\
&\left. - 2 \log_2 \bigg( 2 \bar{p} - 1 + \frac{1}{d} \sqrt{ d (d^2-1) [1-d(2\bar{p} - 1)^2] } \bigg). \right.
\end{split}
\end{equation}
The optimal performance certifies $\log_2 d$ bits of randomness, which is the maximal value for a pair of projective measurements. We plot the above bound for $d = 4$ over the region where it is non-trivial in Fig.~\ref{fig:uncertainty}.
\begin{figure}[h!]
\begin{center}
\includegraphics[width=0.9\columnwidth]{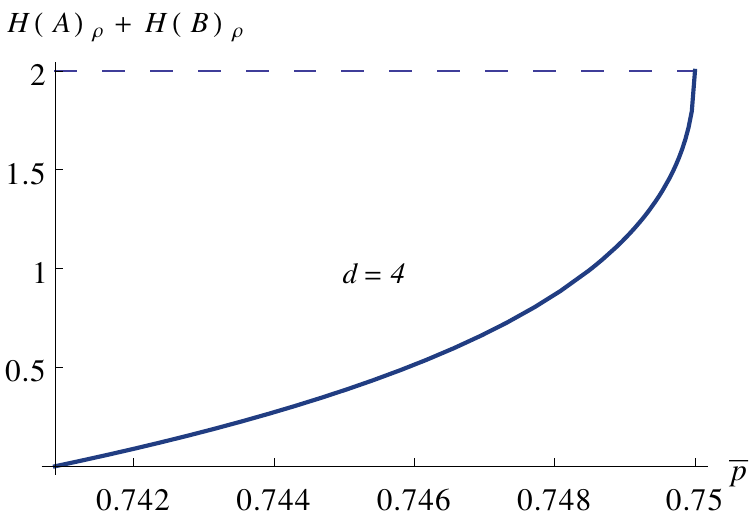}
\caption{Lower bound on the entropic uncertainty over the non-trivial region
in dimension 4.
\label{fig:uncertainty}}
\end{center}
\end{figure}

We note that a similar bound can be derived for the one-shot analogue of
the Shannon entropy, the min-entropy $H_\text{min}$ (which coincides with the
$\infty$-R\'enyi entropy), which is often preferred
in cryptographic scenarios. It was shown in
\cite{minentropy} that for a pair of POVMs, $H_\text{min}(A)_\rho + H_\text{min}
(B)_\rho \ge -\log_2\left(\frac{1+\sqrt{c}}{2}\right)$, for which we can derive a
similar bound to that of \eqref{eq:uncertaintybound}.

\section{Summary and outlook}
We have shown that the $2^{d} \to 1$ QRAC constitutes a robust self-test for MUBs in arbitrary dimension. Observing sufficiently high ASP allows us to deduce that the employed measurements are close to being rank-1 projective and that their overlaps are close to being uniform. The same approach can be used to bound operationally relevant quantities like the incompatibility robustness or the amount of randomness produced. For low dimensions the robustness of our bounds makes them interesting from the experimental point of view.

The most obvious direction for further research is to use our self-testing results to prove SDI security of prepare-and-measure quantum key distribution using high-dimensional systems. One of the main components of the SDI security proof given in~\cite{SDIprot1} is the relation between the observed QRAC performance and the randomness produced for $d = 2$ (qubits). In this work we derive precisely such relations for arbitrary $d$ and we believe that one can use them directly in security proofs.

There is an important difference between SDI self-testing and DI self-testing. In the usual DI self-testing we certify systems up to local isometries and extra degrees of freedom. Since the second equivalence is not relevant in the SDI setup (the dimension of the system is fixed), one might expect that SDI self-testing should characterise the measurements up to a unitary transformation. However, this is generally not the case: while in some dimensions all pairs of MUBs are equivalent up to unitaries (and possibly complex conjugation), e.g.~$d = 2, 3, 5$, there are dimensions where this is not the case, e.g.~$d = 4$~\cite{MUB2-5}. It is natural to ask whether these inequivalent classes of MUBs can be distinguished by considering more complex QRACs.
In fact, a related version of this question appears readily if we consider $n^{d} \to 1$ QRACs with $n > 2$. In this case it is known that different classes of $n$-tuples of MUBs perform differently~\cite{QRACMUB, nUB}. Numerical evidence for $n = 3$ and low $d$ suggests that the optimal performance is achieved by one of these classes, so one might conjecture that such QRACs self-test this particular class. Again, it is not clear how to certify the remaining classes.

The $2^{d} \to 1$ QRAC analysed in this paper is closely related, at least in spirit, to the family of Bell inequalities proposed by Bechmann-Pasquinucci and Gisin~\cite{Gisin}. We hope that the understanding gained in this work will help us to prove self-testing statements for those inequalities. It would be particularly interesting to see whether the need for ``more-than-unitary'' freedom can also appear in the standard nonlocality-based self-testing.

%
%

\section*{Acknowledgements}
We would like to thank Micha{\l} Oszmaniec for fruitful discussions. MF acknowledges support from the Polish NCN grant Sonata~UMO-2014/14/E/ST2/00020. JK acknowledges support from the National Science Centre, Poland (grant no.~2016/23/P/ST2/02122). This project is carried out under POLONEZ programme which has received funding from the European Union's Horizon 2020 research and innovation programme under the Marie Sk{\l}odowska-Curie grant agreement no.~665778.

%

\onecolumngrid

\appendix

\section{Ideal self-test}\label{app:ideal}

In the main text, we establish that the QRAC ASP can be upper bounded by
\begin{equation}\label{eq:appaspbound}
\bar{p}\le\frac{1}{2d^2}\sum_{ij}\norm{A_i+B_j},
\end{equation}
and this can always be saturated by suitable states $\rho_{ij}$ on the
preparation side. In order to bound the above quantity, we use a special case
of a matrix norm inequality derived by Kittaneh \cite{kittaneh}, applied to
the square-root function and the operator norm. For further purposes, we briefly
reproduce the proof here as well.
We will make use of the fact that for operators $A,B$ on a Hilbert space,
$\norm{A\oplus B}=\max\{\norm{A},\norm{B}\}$ \cite{bhatia}.

\begin{thm}\label{thm:Kittaneh}
Let $A,B\ge0$ be operators on a Hilbert space. Then $\norm{A+B}\le
\max\{\norm{A},\norm{B}\}+\norm{\sqrt{A}\sqrt{B}}$.
\end{thm}
\begin{proof}
Consider the block-operator
\begin{equation}
X=\begin{pmatrix}
\sqrt{A} \\
\sqrt{B} 
\end{pmatrix}, \text{ and thus }
X^\dagger X=A+B.
\end{equation}
Therefore
\begin{equation}
\begin{split}
\norm{A+B} & \left. = \norm{X^\dagger X} = \norm{XX^\dagger} = \norm{
\begin{pmatrix}
A & \sqrt{A}\sqrt{B} \\
\sqrt{B}\sqrt{A} & B
\end{pmatrix}
}
= \norm{
\begin{pmatrix}
A & 0 \\
0 & B
\end{pmatrix}+
\begin{pmatrix}
0 & \sqrt{A}\sqrt{B} \\
\sqrt{B}\sqrt{A} & 0
\end{pmatrix}
} \right. \\
& \left. \le \norm{
\begin{pmatrix}
A & 0 \\
0 & B
\end{pmatrix}}+\norm{
\begin{pmatrix}
0 & \sqrt{A}\sqrt{B} \\
\sqrt{B}\sqrt{A} & 0
\end{pmatrix}} 
=  \max\{\norm{A},\norm{B}\}+\norm{\sqrt{A}\sqrt{B}}, \right.
\end{split}
\end{equation}
where we used some basic properties of the operator norm (see
e.g.~\cite{bhatia}; or \cite{kittaneh} for a more detailed and general version
of the proof).
\end{proof}

Using the above theorem, we get

\begin{equation}
\bar{p}\le\frac{1}{2d^2}
\sum_{ij}\Big(\max\{\norm{A_i},\norm{B_j}\} +\norm{\sqrt{A_i}\sqrt{B_j}}\Big).
\end{equation}
From $\sum_i{A_i}=\sum_j{B_j}=\unit$ it follows that $A_i,B_j\le\unit$, and thus
$\norm{A_i},\norm{B_j}\le1$. Then
\begin{equation}
\bar{p}\le\frac{1}{2d^2}\sum_{ij}\Big(1+\norm{\sqrt{A_i
}\sqrt{B_j}}\Big)=\frac12+\frac{1}{2d^2}\sum_{ij}\norm{\sqrt{A_i}\sqrt{B_j}}.
\end{equation}
Now we use the fact that for any operator $O$, $\norm{O}\le\norm{O}_F$, where
$\norm{O}_F := \sqrt{\tr(O^\dagger O)}$ is the Frobenius norm \cite{bhatia}.
Therefore
\begin{equation}\label{eq:tracebound}
\bar{p} \le\frac12+\frac{1}{2d^2}\sum_{ij}\norm{\sqrt{A_i} \sqrt{ B_j}}_F
= \frac12+\frac{1}{2d^2}\sum_{ij}\sqrt{\tr(A_iB_j)}.
\end{equation}
Recall that $t_{ij} := \tr(A_i B_j)$ and, therefore, $t_{ij} \geq 0$ and $\sum_{ij}t_{ij}=d$. The right-hand side of Eq.~\eqref{eq:tracebound} is a symmetric and strictly concave function of
the $t_{ij}$, and as such, it is strictly Schur-concave
(see e.g.~\cite{inequalities}). Therefore, it is maximised {\em uniquely}
by setting all the $t_{ij}$ uniform, $t_{ij}=\frac1d$ for all $i,j\in[d]$. The
upper bound on the ASP set by such $t_{ij}$ is then
\begin{equation}
\bar{p}\le\frac12+\frac{1}{2d^2}\sum_{ij}\frac{1}{\sqrt{d}}=\frac12\Big(1+
\frac{1}{\sqrt{d}}\Big).
\end{equation}
Note that this bound is saturated by measuring in MUBs (see also
\cite{QRACMUB}).

Now, let us turn our attention to necessary conditions for saturating the above
bound. We first show that at least one of the measurements must be rank-1
projective in order to reach the optimal ASP. Saturating the upper bound requires $\tr(A_iB_j)=
\frac1d$ for all $i,j
\in[d]$ and by summing over one of the indices, we see that $\tr A_i
=\tr B_j=1$ for all $i,j$. Investigating the chain of
inequalities obtained above, it is necessary for optimality that 
$\max\{\norm{A_i},\norm{B_j}\}=1$ for all $i,j\in[d]$,
otherwise $\bar{p}<\frac{1}{2d^2}\sum_{ij}(1+\norm{\sqrt{A_i}\sqrt
{B_j}}) \le\frac12\Big(1+\frac{1}{\sqrt{d}}\Big)$.
Assume that there exists a $j^\ast$ such that $\norm{B_{j^\ast}}<1$. Then
in order to fulfil
$\max\{\norm{A_i},\norm{B_{j^{*}}}\}=1$ for all $i \in [d]$, it
is necessary that $\norm{A_i}=1$ for all $i\in[d]$. Since these
operators must all be trace-1 and positive semi-definite, it follows that
$A_i=\ketbraq{a_i}$ for all $i\in[d]$. If there is no such $j^\ast$, then
$\norm{B_j}=1$ for all $j\in[d]$, and we arrive at an analogous condition
for $B_j$. Thus, without loss of generality we can assume that $A_i=\ketbraq{a_i}$ for all $i
\in[d]$.

The rest of this appendix is dedicated to showing that the other measurement must also be rank-1 projective. Let us analyse the inequality derived by Kittaneh and in order to do so,
we first recall a few definitions from matrix analysis. We denote by $\mathcal{L}
(\mathcal{H})$ the algebra of linear operators on the Hilbert space
$\mathcal{H}$, and by $\norm{.}_\mathcal{H}$ the Hilbert space norm. The
numerical range of an operator $O$ is $W(O) := \{\braket{x}{Ox}~|~\norm{x}_
\mathcal{H}=1\}$, while the numerical radius is $w(O) := \sup_ {\norm{x}_
\mathcal{H}=1} \abs{\braket{x}{Ox}}$. By construction every complex number $c \in W(O)$ satisfies $\abs{c} \leq w(O)$ and we always have $w(O)\le\norm{O}$
\cite{bhatia}.

In Theorem~\ref{thm:Kittaneh}, the inequality comes from the triangle inequality and to investigate when this holds as an equality we use a result by Barraa and Boumazgour~\cite{triangle}.

\begin{thm}\label{thm:triangle}
Let $S,T\in\mathcal{L}(\mathcal{H})$ be non-zero. Then the equation $\norm{S+T}
=\norm{S}+\norm{T}$ holds if and only if $\norm{S}\norm{T}\in\overline{W(
S^\dagger T)}$.
\end{thm}
For a finite-dimensional Hilbert space the numerical range is always closed
\cite{bhatia}, thus in our case the closure in the theorem is redundant. It is immediate to see that a necessary condition for the operators $S$
and $T$ to saturate the triangle inequality is that $\norm{S}\norm{T}\le
w(S^\dagger T)$. On the other hand, from the submultiplicativity of the operator norm, we know that $w(S^\dagger T)\le\norm{S^\dagger T}\le\norm{S^\dagger}\norm{T}=\norm{S}
\norm{T}$, and hence this condition is equivalent to $w(S^\dagger T)=\norm{S}
\norm{T}$.

We
will also use the following bound on the numerical radius, obtained by Kittaneh
\cite{kittaneh2}.
\begin{thm}\label{thm:Kittanehradius}
If $O\in\mathcal{L}(\mathcal{H})$, then
\begin{equation}
\big(w(O)\big)^2\le\frac12\norm{O^\dagger O+OO^\dagger}.
\end{equation}
\end{thm}
We are now ready to derive a necessary condition to saturate Kittaneh's inequality in Theorem
\ref{thm:Kittaneh}.
\begin{lem}\label{lem:Kittaneh_sat}
Let $A,B\ge0$ be operators on a Hilbert space. Then, the equality $\norm{A+B}=
\max\{\norm{A},\norm{B}\}+\norm{\sqrt{A}\sqrt{B}}$ holds only if $\norm{A}=\norm{B}$.
\end{lem}
\begin{proof}
Let us denote the block-operators appearing in the proof of Theorem
\ref{thm:Kittaneh} by:
\begin{equation}
S=\begin{pmatrix}
A & 0 \\
0 & B
\end{pmatrix} =S^\dagger,~~~
T=\begin{pmatrix}
0 & \sqrt{A}\sqrt{B} \\
\sqrt{B}\sqrt{A} & 0
\end{pmatrix} =T^\dagger.
\end{equation}
Then, following from Theorem \ref{thm:triangle} and the discussion below
it, a necessary condition for
$A,B\ge0$ to saturate Kittaneh's inequality is that 
$w(ST)=\norm{S}\norm{T}= \max\{\norm{A}, \norm{B}\}\norm{\sqrt{A}
\sqrt{B}}$.

Applying Theorem \ref{thm:Kittanehradius} to $ST$, we get that
\begin{equation}
\begin{split}
(ST)^\dagger ST & \left. =
\begin{pmatrix}
\sqrt{A}B^3\sqrt{A} & 0 \\
0 & \sqrt{B}A^3\sqrt{B}
\end{pmatrix}, \right. \\
ST(ST)^\dagger & \left. =
\begin{pmatrix}
A^{\frac32}BA^{\frac32} & 0 \\
0 & B^{\frac32}AB^{\frac32}
\end{pmatrix}, \right.
\end{split}
\end{equation}
and hence
\begin{equation}
\begin{split}
\big(w(ST)\big)^2 & \left.
\le\frac12\max\Big\{\norm{\sqrt{A}B^3\sqrt{A}
+A^{\frac32}BA^{\frac32}},\norm{\sqrt{B}A^3\sqrt{B}+B^{\frac32}
AB^{\frac32}}\Big\} \right. \\
& \left. \le \frac12\max\Big\{\norm{\sqrt{A}B^3\sqrt{A}}+\norm{A^{\frac32}B
A^{\frac32}},\norm{\sqrt{B}A^3\sqrt{B}}+\norm{B^{\frac32}AB^{\frac32}}\Big\}
\right. \\
& \left. = \frac12\max\Big\{\norm{\sqrt{A}B^{\frac32}}^2+\norm{A^{\frac32}
\sqrt{B}}^2,\norm{A^{\frac32}\sqrt{B}}^2+\norm{\sqrt{A}B^{\frac32}}^2\Big\}
\right. \\
& \left. = \frac12\Big(\norm{A^{\frac32}\sqrt{B}}^2+\norm{\sqrt{A}B^{\frac32}}^2
\Big)\le\frac12\Big(\norm{A}^2+\norm{B}^2\Big)\norm{\sqrt{A}\sqrt{B}}^2
\right. \\
& \left. \le\max\big\{\norm{A}^2,\norm{B}^2\big\}\norm{\sqrt{A}\sqrt{B}}^2.
\right.
\end{split}
\end{equation}

Here, in the second line, we used the triangle inequality, in the third line 
the identity $\norm{O}^{2}=\norm{O^\dagger O}$ and in the fourth line submultiplicativity. The last inequality is trivial, and is only saturated if
$\norm{A}=\norm{B}$. Therefore, $\norm{A+B}=
\max\{\norm{A},\norm{B}\}+\norm{\sqrt{A}\sqrt{B}}$ only if $\norm{A}=\norm{B}$.
\end{proof}

This lemma shows that saturating the upper bound on the ASP implies that $\norm{B_j}=\norm{A_i} = 1$ for all $i,j\in[d]$. It was also necessary that $\tr B_j=1$, and
therefore (similarly to the $A_i$), $B_j=\ketbraq{b_j}$ for all $j\in[d]$,
and both measurements must be rank-1 projective. From here, it follows
immediately from the condition $\tr(A_iB_j)=\frac1d$, that the bases
defining the measurements must be mutually unbiased.

\section{Robust self-test}
\label{app:robust}
While it is clear what it means for two measurements to be \emph{exactly} mutually unbiased, there are multiple ways of turning this definition into an approximate statement (particularly if we allow for non-projective measurements). For our purposes it is natural to split the definition of MUBs into two stand-alone conditions and consider them separately.

The first condition, which is usually implicit in the definition of MUBs, is that both measurements are projective and that the measurement operators are rank-1. Let $\{ A_{i} \}_{i}$ be a $d$-outcome measurement on a $d$-dimensional system and let us consider the sum of the norms, $N(A) := \sum_{i} \norm{ A_{i} }$. This is a suitable quantity, because
\begin{equation*}
N(A) = \sum_{i} \norm{ A_{i} } \leq \sum_{i} \tr A_{i} = d
\end{equation*}
and since $\norm{ A_{i} } \leq 1$, the maximum is achieved iff every measurement operator is a rank-1 projector. Therefore, the difference between $\sum_{i} \norm{ A_{i} }$ and the maximal value $d$ tells us how much $\{ A_{i} \}_{i}$ deviates from being rank-1 projective.

The second condition, often referred to as \emph{the} MUB condition, requires that the overlap between every pair of measurement operators is the same. The question here is how to generalise the overlap to non-projective measurements. The quantity $\sqrt{ \tr ( A_{i} B_{j} ) }$ discussed in the main text is a valid generalisation of the overlap in the sense that it reduces to the overlap for rank-1 projective measurements. However, the argument given below naturally leads to a different quantity, namely $\norm{ \sqrt{A_{i}} \sqrt{B_{j}} }$. Note that this is a commonly used definition of the overlap, e.g.~in the context of uncertainty relations.

The main purpose of this appendix is to derive a lower bound on $N(A)$ as a function of the observed performance. However, in order to do that, we must first derive explicit bounds on the range of $\norm{ \sqrt{A_{i}} \sqrt{B_{j}} }$.

%
%
%
%
%
In our argument we use the following technical lemma.
\begin{lem}
\label{eq:random-inequality}
The function
\begin{equation*}
h(x, y) := x + y - \alpha x y - \sqrt{ x^{2} + y^{2} }
\end{equation*}
for $\alpha := 2 - \sqrt{2}$ satisfies $h(x, y) \geq 0$ for $x, y \in [0, 1]$.
\end{lem}
\begin{proof}
If we express $x$ and $y$ in terms of the polar coordinates
\begin{align*}
x &= r \cos ( \theta - \pi/4 ),\\
y &= r \sin ( \theta - \pi/4 ),
\end{align*}
the function becomes
\begin{equation*}
h(r, \theta) = r \big[ \cos ( \theta - \pi/4 ) + \sin ( \theta - \pi/4 ) - 1 \big] - \frac{\alpha r^{2}}{2} \sin \big[ 2 ( \theta - \pi/4 ) \big] = r \big( \sqrt{2} \sin \theta - 1 \big) + \frac{\alpha r^{2}}{2} \cos 2 \theta.
\end{equation*}
To cover the square $x, y \in [0, 1]$ we prove the statement for $r \in [ 0, \sqrt{2} ]$ and $\theta \in [\pi/4, 3 \pi/4]$. For fixed $\theta$ the function $h(r, \theta)$ is a quadratic function of $r$ and the coefficient of the quadratic term is non-positive. This means that in order to determine the minimum value, it suffices to consider the extreme points, i.e.~$r = 0$ and $r = \sqrt{2}$. Since $h(0, \theta) = 0$, we only have to look at the latter. We have
\begin{equation*}
h(\sqrt{2}, \theta) = 2 \sin \theta - \sqrt{2} + \alpha \cos 2 \theta = - 2 \alpha \sin^{2} \! \theta + 2 \sin \theta + 2 - 2 \sqrt{2} = 2 \alpha (1 - \sin \theta) \bigg( \sin \theta - \frac{1}{\sqrt{2}} \bigg)
\end{equation*}
and it is easy to see that for $\theta \in [\pi/4, 3 \pi/4]$ each term is non-negative.
\end{proof}
%
Moreover, we use the following operator norm inequality derived by Kittaneh~\cite{Kittanehsum}.
\begin{thm}
\label{thm:Kittaneh2}
For positive semidefinite operators $A$ and $B$ acting on a finite-dimensional Hilbert space we have
\begin{equation}
\label{eq:kittaneh2}
\norm{ A + B } \leq \frac{1}{2} \left( \norm{A} + \norm{B} + \sqrt{\left( \norm{A} -\norm{B} \right)^{2} + 4 \norm{ \sqrt{A} \sqrt{B} }^2 } \right).
\end{equation}
\end{thm}
In our argument $A$ and $B$ will be particular measurement operators from the two measurements. We define the \emph{generalised overlap} between $A_{i}$ and $B_{j}$ as
\begin{equation*}
s_{ij} := \norm{ \sqrt{A_{i} } \sqrt{ B_{j} } } \in [0, 1].
\end{equation*}
Another relevant quantity of a pair of measurement operators is the \emph{norm deficiency} defined as
\begin{equation*}
n_{ij} := 1 - ( \norm{ A_{i} } + \norm{ B_{j} } )/2 \in [0, 1].
\end{equation*}
It is easy to see that if $n_{ij} = 0$ for all $i, j$, we have
\begin{equation*}
\sum_{i} \norm{ A_{i} } = \sum_{j} \norm{ B_{j} } = d,
\end{equation*}
i.e.~both measurements are rank-1 projective. Our goal now is to relate the right-hand side of Eq.~\eqref{eq:kittaneh2} to $s_{ij}$ and $n_{ij}$. First, note that
\begin{equation*}
\norm{ A_{i} } - \norm{ B_{j} } = 2 \norm{ A_{i} } - ( \norm{ A_{i} } + \norm{ B_{j} } ) \leq 2 - 2 (1 - n_{ij}) = 2 n_{ij}
\end{equation*}
and similarly
\begin{equation*}
\norm{ B_{j} } - \norm{ A_{i} } \leq 2 n_{ij}.
\end{equation*}
These two inequalities imply that
\begin{equation*}
\big( \norm{ A_{i} } - \norm{ B_{j} } \big)^{2} \leq 4 n_{ij}^{2}
\end{equation*}
and plugging this back into Eq.~\eqref{eq:kittaneh2} gives
\begin{equation*}
%
\norm{ A_{i} + B_{j} } \leq 1 - n_{ij} + \sqrt{ n_{ij}^{2} + s_{ij}^{2} }.
\end{equation*}
Applying the inequality derived in Lemma~\ref{eq:random-inequality} to $s_{ij}$ and $n_{ij}$ gives
\begin{equation*}
\norm{ A_{i} + B_{j} } \leq 1 + s_{ij} - \alpha s_{ij} n_{ij},
\end{equation*}
where $\alpha = 2 - \sqrt{2}$. Applying this upper bound to Eq.~\eqref{eq:appaspbound} immediately yields
\begin{equation}
\label{eq:appaspboundns}
\bar{p} \leq \frac{1}{2 d^{2}} \sum_{ij} \big( 1 + s_{ij} - \alpha s_{ij} n_{ij} \big) = \frac{1}{2} + \frac{1}{2 d^{2}} \sum_{ij} s_{ij} - \frac{ \alpha }{ 2 d^{2} } \sum_{ij} s_{ij} n_{ij}.
\end{equation}
Let us first bound the range of $s_{ij}$, i.e.~find explicit functions of $\bar{p}$ denoted by $\smin$ and $\smax$ such that
\begin{equation*}
s_{ij} \in [\smin, \smax]
\end{equation*}
for all $i, j$. To do this we drop the last term in Eq.~\eqref{eq:appaspboundns} to obtain
\begin{equation*}
\bar{p} \leq \frac{1}{2} + \frac{1}{2 d^{2}} \sum_{ij} s_{ij}.
\end{equation*}
To bound the sum of $s_{ij}$ we bound the operator norm by the Frobenius norm:
\begin{equation*}
s_{ij} = \norm{ \sqrt{ A_{i} } \sqrt{ B_{j} } } \leq \norm{ \sqrt{ A_{i} } \sqrt{ B_{j} } }_{F} = \sqrt{ \tr ( A_{i} B_{j} ) } = \sqrt{ t_{ij} }
\end{equation*}
and finally use the normalisation condition $\sum_{ij} t_{ij} = d$. Let us now separate one term from the rest of the sum. For simplicity we choose the first term, i.e.~$s_{11}$, but by symmetry the same argument applies to every $s_{ij}$. We obtain
\begin{equation}
\label{eq:barp}
\bar{p} \leq \frac{1}{2} + \frac{1}{2 d^{2}} \Big( s_{11} + \sum_{ij \neq 11} s_{ij} \Big) \leq \frac{1}{2} + \frac{1}{2 d^{2}} \Big( s_{11} + \sum_{ij \neq 11} \sqrt{ t_{ij} } \Big).
\end{equation}
Since the remaining sum contains $d^{2} - 1$ terms, concavity of the square root implies that
\begin{equation*}
\sum_{ij \neq 11} \frac{1}{ d^{2} - 1 } \sqrt{ t_{ij} } \leq \sqrt{ \frac{ \sum_{ij \neq 11} t_{ij} }{ d^{2} - 1 } } = \sqrt{ \frac{ d - t_{11} }{ d^{2} - 1 } } \leq \sqrt{ \frac{ d - s_{11}^{2} }{ d^{2} - 1 } },
\end{equation*}
where in the last step we used the fact that $s_{11} \leq \sqrt{ t_{11} }$. Plugging this bound into Eq.~\eqref{eq:barp} gives
\begin{equation*}
\bar{p} \leq \frac{1}{2} + \frac{1}{2 d^{2}} \Big( s_{11} + \sqrt{ ( d^{2} - 1 )( d - s_{11}^{2} ) } \Big) =: f( s_{11} ).
\end{equation*}
Computing the derivative of $f$ shows that $f$ is increasing for $s_{11} < 1/\sqrt{d}$ and decreasing for $s_{11} > 1/\sqrt{d}$. The maximum achieved for $s_{11} = 1/\sqrt{d}$ corresponds to the optimal ASP. This implies that the lowest and highest values of $s_{11}$ compatible with the observed $\bar{p}$ can be determined by computing the two solutions of the equality
\begin{equation*}
\bar{p} = \frac{1}{2} + \frac{1}{2 d^{2}} \Big( s_{11} + \sqrt{ ( d^{2} - 1 )( d - s_{11}^{2} ) } \Big).
\end{equation*}
This reduces to solving a quadratic equation and finally we deduce that $s_{11} \in [\smin, \smax]$, where
\begin{align}
\label{eq:smin}
\smin := 2 \bar{p} - 1 - \frac{1}{d} \sqrt{ d ( d^{2} - 1 ) [ 1 - d ( 2 \bar{p} - 1 )^{2} ] },\\
\label{eq:smax}
\smax := 2 \bar{p} - 1 + \frac{1}{d} \sqrt{ d ( d^{2} - 1 ) [ 1 - d ( 2 \bar{p} - 1 )^{2} ] }.
\end{align}
The optimal performance, i.e.~$\bar{p} = \frac{1}{2} + \frac{1}{2 \sqrt{d}}$, implies that $\smin = \smax = \frac{1}{\sqrt{d}}$. Moreover, since both functions are continuous in $\bar{p}$, for sufficiently good performance we obtain bounds stronger than the trivial $s_{11} \geq 0$ and $s_{11} \leq 1$. This concludes the first part of the argument, i.e.~providing explicit bounds on the range of the generalised overlaps.

For the second part of the argument, in which we show that the measurements are close to being rank-1 projective, we need all the overlaps to be bounded away from $0$, i.e.~$\smin > 0$. According to Eq.~\eqref{eq:smin} this is guaranteed as long as $\bar{p} > \bar{p}_{0}$ for
\begin{equation*}
\bar{p}_{0} := \frac{1}{2} + \frac{1}{2 d^{2}} \sqrt{ (d^{2} - 1) d }.
\end{equation*}
Using the concavity result while keeping the negative term in Eq.~\eqref{eq:appaspboundns} leads to
\begin{equation*}
\bar{p} \leq \frac{1}{2} + \frac{1}{2 d^{2}} \Big( s_{11} + \sqrt{ ( d^{2} - 1 )( d - s_{11}^{2} ) } \Big) - \frac{ \alpha }{ 2 d^{2} } \sum_{ij} s_{ij} n_{ij}.
\end{equation*}
Without loss of generality we can assume that $s_{11}$ is the smallest overlap and then
\begin{equation*}
\bar{p} \leq \frac{1}{2} + \frac{1}{2 d^{2}} \Big( s_{11} + \sqrt{ ( d^{2} - 1 )( d - s_{11}^{2} ) } \Big) - \frac{ \alpha s_{11} }{ 2 d^{2} } \sum_{ij} n_{ij},
\end{equation*}
which is equivalent to
\begin{equation}
\label{eq:sum-nij}
\sum_{ij} n_{ij} \leq \frac{1}{ \alpha s_{11} } \bigg( s_{11} + \sqrt{ ( d^{2} - 1 )( d - s_{11}^{2} ) } - d^{2} ( 2 \bar{p} - 1 ) \bigg).
\end{equation}
To analyse the right-hand side, we define
\begin{equation*}
g(x) := 1 + \sqrt{ ( d^{2} - 1 ) \bigg( \frac{d}{x^{2}} - 1 \bigg) } - \frac{ d^{2} ( 2 \bar{p} - 1 ) }{x},
\end{equation*}
and now our goal is to maximise $g(x)$ over $x \in [0, 1/\sqrt{d}]$, as
$s_{\text{min}}\le1/\sqrt{d}$. Recall that we work under the assumption that $\bar{p} > \bar{p}_{0}$ and therefore $2 \bar{p} - 1 > 0$. We can analytically compute the derivative $dg/dx$ and set it to $0$ to conclude that the only stationary point corresponds to
\begin{equation*}
x^{*} := \frac{ \sqrt{ d^{3} ( 2 \bar{p} - 1 )^{2} - ( d^{2} - 1 ) } }{ d (2 \bar{p} - 1 ) } = \sqrt{ d - \frac{ d^{2} - 1 }{ d^{2} ( 2 \bar{p} - 1 )^{2} } }.
\end{equation*}
Evaluating the second derivative $d^{2} g / d x^{2}$ at $x^{*}$ tells us that this is a maximum and since this is the only stationary point, it must be the unique maximiser in the interval $[0, 1/\sqrt{d}]$.
Therefore, in Eq.~\eqref{eq:sum-nij} we can set $s_{11} = x^{*}$ to obtain
\begin{equation*}
\sum_{ij} n_{ij} \leq \frac{1}{\alpha} \Big( 1 - \sqrt{ d^{3} ( 2 \bar{p} - 1 )^{2} - ( d^{2} - 1 ) } \Big).
\end{equation*}
Finally, we can use this bound to obtain lower bounds on the sums of the norms $\sum_{i} \norm{A_{i}}$ and $\sum_{j} \norm{ B_{j} }$ for the individual measurements. Since
\begin{equation*}
\sum_{ij} n_{ij} = d^{2} - \frac{d}{2} \Big( \sum_{i} \norm{ A_{i} } + \sum_{j} \norm{ B_{j} } \Big),
\end{equation*}
we can use the trivial bound $N(B) = \sum_{j} \norm{ B_{j} } \leq d$ to obtain
\begin{equation}
\label{eq:sum-of-norms}
N(A) = \sum_{i} \norm{ A_{i} } \geq d - \frac{2}{d} \sum_{ij} n_{ij} \geq d - \frac{2}{\alpha d} \Big( 1 - \sqrt{ d^{3} ( 2 \bar{p} - 1 )^{2} - ( d^{2} - 1 ) } \Big).
\end{equation}
Clearly, the same lower bound holds for $N(B)$.

\section{Incompatibility robustness}
\label{app:incomp}
In this appendix we derive an analytic upper bound on the incompatibility robustness as a function of the observed ASP. We start with a bound derived recently in~\cite{incompbound}:
\begin{equation}
\label{eq:appincompbound}
\eta^{\ast} \leq \frac{ d^{2} \max_{ij} \norm{ A_{i} + B_{j} } - \sum_{i} ( \tr A_{i} )^{2} - \sum_{j} ( \tr B_{j} )^{2} }{ d \sum_{i} \tr A_{i}^{2} + d \sum_{j} \tr B_{j}^{2} - \sum_{i} ( \tr A_{i} )^{2} - \sum_{j} ( \tr B_{j} )^{2} }.
\end{equation}
The aim is to bound all the terms appearing in this formula by quantities which we have already bounded in Appendix~\ref{app:robust}.

Let us start with the numerator. The first term is easy to bound since
\begin{equation*}
\norm{ A_{i} + B_{j} } \leq 1 +s_{ij},
\end{equation*}
and $\max_{ij} s_{ij} \leq \smax$ given in Eq.~\eqref{eq:smax}.

To bound the second term we use the fact that for positive semidefinite operators $(\tr A)^{2} \geq \tr A^{2}$ and then bound the Frobenius norm by the operator norm:
\begin{equation*}
( \tr A_{i} )^{2} \geq \tr A_{i}^{2} = \norm{ A_{i} }_{F}^{2} \geq \norm{ A_{i} }^{2}.
\end{equation*}
To bound the sum of the squares $\sum_{i} \norm{ A_{i} }^{2}$ we use a standard inequality for vector $p$-norms which for $d$-dimensional vectors reads $\norm{x}_{2} \geq \frac{1}{\sqrt{d}} \norm{x}_{1}$. Applying this to the real vector whose components are given by $x_{i} = \norm{ A_{i} }$ yields
\begin{equation*}
\sum_{i} \norm{ A_{i} }^{2} \geq \frac{1}{d} \Big( \sum_{i} \norm{ A_{i} } \Big)^{2}.
\end{equation*}
Putting the two inequalities together gives
\begin{equation*}
\sum_{i} ( \tr A_{i} )^{2} \geq \frac{1}{d} \Big( \sum_{i} \norm{ A_{i} } \Big)^{2},
\end{equation*}
which can be bounded using Eq.~\eqref{eq:sum-of-norms}.

The first term in the denominator we have already bounded: from the previous argument we see that
\begin{equation*}
\sum_{i} \tr A_{i}^{2} \geq \frac{1}{d} \Big( \sum_{i} \norm{ A_{i} } \Big)^{2}.
\end{equation*}
Bounding the last term turns out to be slightly more involved, so we state it as a separate lemma.

\begin{lem}
\label{lem:appsumtracesquarebound}
Let $\{ A_{i} \}_{i}$ be a $d$-outcome measurement acting on $\mathbb{C}^{d}$. If
\begin{equation*}
\sum_{i} \norm{ A_{i} } \geq q,
\end{equation*}
then
\begin{equation*}
\sum_{i} ( \tr A_{i} )^{2} \leq d + ( d - q ) ( d - q + 1 ).
\end{equation*}
\end{lem}
\begin{proof}
Before proceeding to the technical details, let us briefly explain the idea behind the proof. Suppose we are given a partition of the $d$ measurement outcomes into two disjoint sets. Moreover, we are promised that the trace of the measurement operators corresponding to the outcomes in the first (second) set belongs to the interval $[0, 1]$ ($[1, d]$). It turns out that an upper bound on the desired quantity can be derived in terms of simple properties of this partition. Maximising this bound over all valid partitions leads to the main result of the lemma.

Formally, we are given two sets $X$ and $Y$ such that $X \cup Y = [d]$ and $X \cap Y = \emptyset$. Moreover, we have
\begin{align*}
i \in X &\implies \tr A_{i} \in [0, 1],\\
i \in Y &\implies \tr A_{i} \in [1, d].
\end{align*}
Define $n := \abs{X}$, $\gamma := \sum_{i \in X} \tr A_{i}$ and clearly
\begin{equation}
\label{eq:constraint-1}
n - \gamma \geq 0.
\end{equation}
Moreover, the assumption of the lemma implies
\begin{equation*}
q \leq \sum_{i} \norm{ A_{i} } = \sum_{i \in X}\norm{A_i} + \sum_{i \in Y}\norm{A_i} \leq  \sum_{i \in X} \tr A_{i} + \abs{Y}  = \gamma + d - n
\end{equation*}
and therefore
\begin{equation}
\label{eq:constraint-2}
n - \gamma \leq d - q.
\end{equation}
For the rest of the argument let us think of $n$ and $\gamma$ as some fixed values. Once we derive the final upper bound in terms of these two variables, we will maximise it over the allowed pairs of $n$ and $\gamma$.

For $i \in X$ we have $(\tr A_{i})^{2} \leq \tr A_{i}$ and therefore
\begin{equation*}
\sum_{i \in X} ( \tr A_{i} )^{2} \leq \sum_{i \in X} \tr A_{i} = \gamma.
\end{equation*}
To bound the second term we must explicitly determine the allowed combinations of $\{ \tr A_{i} \}_{i \in Y}$. Since $\{ \tr A_{i} \}_{i \in Y} \in [1, d]^{\abs{Y}}$ and
\begin{equation*}
\sum_{i \in Y} \tr A_{i} = d - \gamma,
\end{equation*}
the valid choices of $\{ \tr A_{i} \}_{i \in Y}$ form a polytope. It is easy to see that all the vertices of this polytope correspond to setting $\abs{Y} - 1$ values to $1$ and the last value to $[d - \gamma - ( \abs{Y} - 1 ) ]$. Since $\sum_{i \in Y} ( \tr A_{i} )^{2}$ is a convex function of the traces, the maximal value is achieved at a vertex and therefore
\begin{equation*}
\sum_{i \in Y} ( \tr A_{i} )^{2} \leq (\abs{Y} - 1) + \big[ d - \gamma - (\abs{Y} - 1) \big]^{2}.
%
\end{equation*}
Plugging in $\abs{Y} = d - n$ gives
\begin{equation*}
\sum_{i \in Y} ( \tr A_{i} )^{2} \leq d - n - 1 + ( n - \gamma + 1 )^{2} = d + ( n - \gamma ) ( n - \gamma + 1 ) - \gamma.
\end{equation*}
Putting the two bounds together leads to
\begin{equation*}
\sum_{i} ( \tr A_{i} )^{2} = \sum_{i \in X} ( \tr A_{i} )^{2} + \sum_{i \in Y} ( \tr A_{i} )^{2} \leq d + ( n - \gamma ) ( n - \gamma + 1 ).
\end{equation*}
Now we must maximise the right-hand side subject to the constraints given in Eqs.~\eqref{eq:constraint-1} and~\eqref{eq:constraint-2}. The maximum is achieved when the latter is saturated, which leads to the final result of the lemma.
\end{proof}
The final bound reads
\begin{equation}
\eta^{\ast} \leq \frac{ \frac{1}{2} d^{2} ( 1 + \smax ) - \frac{q^{2}}{d} }{ q^{2} - d - (d - q)(d - q + 1) }.
\end{equation}
where $\smax$ is the quantity defined in Eq.~\eqref{eq:smax}, while $q$ is the right-hand side of Eq.~\eqref{eq:sum-of-norms}.
\twocolumngrid

%

%

\end{document}